\documentclass[a4paper,UKenglish,cleveref,autoref,thm-restate]{lipics-v2021}

\usepackage{mathtools}
\usepackage[ruled,vlined]{algorithm2e}

\usepackage[noEnd,commentColor=black]{algpseudocodex}
\usepackage{graphicx}
\usepackage{subcaption}
\usepackage{thmtools}
\usepackage{todonotes}
\usepackage{microtype}

\hideLIPIcs

\usepackage[normalem]{ulem}
\usepackage{tcolorbox}
\usepackage{upgreek}
\usepackage{cancel}

\crefname{equation}{eq.}{eqs.} 
\crefname{enumi}{}{} 

\crefname{icase}{case}{cases}
\crefname{ipart}{part}{parts}
\crefname{iprop}{property}{properties}
\crefname{iinv}{invariant}{invariants}

\crefformat{section}{\S\,#2#1#3}
\crefformat{subsection}{\S\,#2#1#3}
\crefrangeformat{section}{\S\S\,#3#1#4--#5#2#6}
\crefmultiformat{section}{\S\S\,#2#1#3}{,~#2#1#3}{,~#2#1#3}{,~#2#1#3}

\usepackage{tikz}
\usetikzlibrary{calc}

\newcommand{\eat}[1]{}

\newcommand{\cA}{\mathscr{A}}

\DeclareMathOperator{\replace}{\mathsf{replace}}
\DeclareMathOperator{\decode}{\mathsf{decode}}
\DeclareMathOperator{\freq}{\mathsf{freq}}

\DeclareMathOperator{\OPT}{\mathsf{OPT}}
\DeclareMathOperator{\bpe}{\mathsf{BPE}}

\renewcommand{\alpha}{\upalpha}

\newcommand{\cP}{\mathscr{P}}

\newcommand{\cR}{\mathscr{R}}

\newcommand{\cB}{\mathscr{B}}

\renewcommand{\alpha}{\upalpha}

\newcommand{\connected}[1]{\def\temp{#1}\ifx\temp\empty\sim\else\overset{#1}{\sim}\fi}

\tikzset{
	point/.style={circle, fill, inner sep=1.5pt},
	smallpoint/.style={point, inner sep=1.2pt},
	tinypoint/.style={point, inner sep=1pt},
	hlbox/.style={fill, {white!90!black}},
	subrect/.style={draw, fill={white!80!cyan}}, 
	msrect/.style=subrect 
}

\nolinenumbers

\title{Theoretical Analysis of Byte-Pair Encoding}

\author{L\'aszl\'o Kozma}{Faculty of Computer Science, TU Dresden, Germany}{laszlo.kozma@tu-dresden.de}{}{}

\author{Johannes Voderholzer}{TomTom}{jvoderho@gmail.com}{}{}

\acknowledgements{Work partly done while both authors were at Freie Universität Berlin.}

\authorrunning{L.~Kozma and J.~Voderholzer}

\renewcommand{\keywordsHeading}{}

\keywords{\mbox{}}

\category{}

\renewcommand{\subjclassHeading}{}

\ccsdesc[100]{\hspace{0pt}}

\Copyright{}

\EventEditors{}
\EventLongTitle{}
\EventShortTitle{}
\EventAcronym{}
\EventYear{}
\EventDate{}
\EventLocation{}
\EventLogo{}
\SeriesVolume{}
\ArticleNo{}

\begin{document}

\maketitle

\begin{abstract}
Byte-Pair Encoding (BPE) is a widely used method for subword tokenization, with origins in grammar-based text compression. It is employed in a variety of language processing tasks such as machine translation or large language model (LLM) pretraining, to create a token dictionary of a prescribed size.  Most evaluations of BPE to date are empirical, and the reasons for its good practical performance are not well understood. 

In this paper we focus on the optimization problem underlying BPE: finding a \emph{pair encoding} that achieves optimal compression utility. We show that this problem is APX-complete, indicating that it is unlikely to admit  a polynomial-time approximation scheme.  This answers, in a stronger form, a question recently raised by Zouhar et al.~\cite{zouhar2023formal}. 

On the positive side, we show that BPE approximates the compression utility of the optimal pair encoding to a worst-case factor between $1/3$ and $(1-e^{-2})/2 \approx 0.4324$. Our results aim to explain the ongoing success of BPE and are, to our knowledge, the first rigorous guarantees on its compression utility that hold for all inputs.\footnote{This is an extended version of a conference paper appearing at ESA 2026. Compared to the proceedings version, it includes an improved approximation upper bound ($0.4324$ instead of $0.625$) and fixes an inaccuracy in the hardness proof of \S\,\ref{sec2}.}

\end{abstract}

\section{Introduction}\label{sec1}

A common step in the modern natural language processing (NLP) application pipeline is \emph{tokenization}: given an input text, the task is to partition it into \emph{tokens}, i.e., frequently occurring consecutive groups of symbols. The main goal is to identify  semantically meaningful units (words or subwords) in order to facilitate higher level tasks (e.g., translation or text generation)~\cite{mielke2021between, ali2023tokenizer}. As this goal is difficult to directly optimize for, tokenization is usually solved heuristically, or formulated as a different but closely related task: \emph{data compression}. 
Indeed, the dictionary-encoding of tokens reduces text length; the amount of compression is easy to measure, and was found to be a good predictor of the quality of tokenization for downstream tasks, e.g., for translation accuracy~\cite{galle2019investigating} or language model performance~\cite{goldman2024unpackingtokenizationevaluatingtext}. It is thus natural to study tokenization with the proxy optimization goal of \emph{compression utility}.

Byte-Pair Encoding (BPE), introduced by Gage in 1994~\cite{gage1994new}, is a commonly used heuristic for tokenization. It proceeds by repeatedly identifying the \emph{most frequently occurring} pair of symbols and replacing all occurrences of this pair with a new symbol, thereby shortening the text.\footnote{The special cases of overlapping copies of a pair and ties in frequency will be discussed later.} The new symbols, together with the pairs they replace, are stored in a lookup-table, which allows the reconstruction of the original text. 
In typical applications, the number of new symbols (and thus the size of the lookup-table) is fixed upfront, and the goal is to achieve the best compression within this budget. The symbols of the resulting (shortened) text correspond to the tokens of the input. Figure~\ref{fig1} shows an example of the encoding of a text by BPE.

BPE has become a de-facto standard in NLP applications, widely employed in machine translation~\cite{sennrich2015neural, xu2021vocabulary, domingo2019much, galle2019investigating, gutierrez2023languages}
and in the preprocessing stage of training large language models~\cite{BrownMRSKDNSSAA20, liu2021robustly,touvron2023llama,radford2019language,le2023bloom, wu2023multimodal, liu2024deepseek}\footnote{See also  \url{https://github.com/google/sentencepiece} ~ and ~ \url{ https://github.com/openai/tiktoken}.}. 
Besides its effectiveness, the popularity of BPE is likely due to its simplicity and computational efficiency, when compared with more sophisticated or linguistically motivated methods (e.g., see~\cite{bostrom2020byte,schmidt2024tokenization}). A careful implementation of BPE has a total runtime that is linear in the input length, for an arbitrary number of replacement rounds. In addition, a BPE-encoded representation can support efficient random access and pattern-matching on the original text, which is important in some applications~\cite{shibata1999byte}.

Given the popularity and good empirical performance of BPE, there is a surprising lack of rigorous guarantees for the quality of its output. In this paper we study the problem of compressing a text (a string $s$ over some alphabet $\Sigma$) by successive encoding of pairs (strings of length two). Adopting the framework of approximation algorithms~\cite{approx_book} we study how well BPE, as a natural greedy heuristic, approximates this problem.  
Our optimization goal is to maximize \emph{compression utility}, i.e., the reduction in text length, within $k$ pair-encoding rounds, where $s$ and $k$ are given as the input (we precisely define this problem -- \emph{optimal pair encoding} -- later in this section). In applications of BPE the value of $k$ is informed both by domain knowledge (the size of the dictionary we aim to meaningfully extract) and system constraints (the afforded memory footprint), it is thus natural to fix $k$ upfront. 

Our problem definition closely resembles the one recently introduced by Zouhar et al.~\cite{zouhar2023formal} for the same task. 
This abstract setting presents a challenging algorithm design problem and allows a clean theoretical analysis of BPE. Note however, that we necessarily ignore some practical aspects and optimizations of BPE-variants (e.g., special treatment of whitespace and punctuation or language-specific rules~\cite{radford2019language, ali2023tokenizer}). 

An algorithm $\cA$ for optimal pair encoding has \emph{approximation ratio} $\alpha \leq 1$, if the compression utility of $\cA$ is at least $\alpha$ times the optimum for all inputs $(s,k)$. 
The greedy step of BPE is locally optimal, and thus, for $k=1$ it achieves optimal compression. For $k>1$, however, simple examples show that BPE may not achieve optimal compression (see Figure~\ref{fig1}). 

As our main complexity-result, we show that optimal pair encoding is \emph{APX-complete}. This means (informally) that no polynomial-time algorithm can approximate it to a factor arbitrarily close to $1$, unless P=NP. On the positive side, we show that BPE achieves an approximation ratio $\alpha$, with $0.3333 \leq \alpha \leq 0.4324$. 
We note that previously no constant-approximation guarantee was known for BPE or other algorithms. The question of whether optimal pair encoding is NP-complete was raised recently by Zouhar et al.~\cite{zouhar2023formal}; our result settles this question in a stronger form.

\begin{figure} 
\captionsetup{width=.92\linewidth}
\begin{align*}
\quad \quad \quad & \mathtt{aabaaaba \rightarrow XbXaba \rightarrow YXaba \rightarrow Zaba}\\
\quad \quad \quad & \mathtt{aabaaaba \rightarrow aXaaXa \rightarrow YaYa \rightarrow ZZ}
\end{align*}
\caption{Input $s = \mathtt{aabaaaba}$ encoded by BPE merge sequence $(\mathtt{aa \rightarrow X,\, Xb \rightarrow Y,\, YX \rightarrow Z)}$ (above). An optimal encoding by the merge sequence $(\mathtt{ab \rightarrow X,\, aX \rightarrow Y,\, Ya \rightarrow Z})$ (below). Here, the difference is due to an unlucky tie-breaking; we show later that BPE can be far from optimal even if there are no ties.\label{fig1} }
\end{figure}

Before precisely stating our results, we review some further related work and give a formal definition of the problem and the algorithms that we study.

\subparagraph{Related work.} BPE has its origins in text compression, in particular, the class of  \emph{grammar-based} compression methods or \emph{macro schemes}, e.g., see~\cite{storer1982data, kieffer2000grammar, CharikarLLPPSS05, lohrey2012algorithmics} for surveys. (The encoding obtained by BPE can be seen as a restricted kind of context-free grammar or string straight-line program.) 
A compression method closely related to BPE is Re-Pair~\cite{larsson2000off}. Re-Pair differs from BPE in that its number of replacement rounds $k$ is not fixed; instead, it performs replacements as long as they achieve a saving (i.e., as long as some pair appears in at least two disjoint copies). Re-Pair is widely used, e.g., in bioinformatics, and several variants and practical improvements of it have been proposed~\cite{kim2024recursive, ganczorz2017improvements, gagie2019rpair, MRRepair}.

The central question of grammar-based compression is to find a minimal grammar that generates a given text. This task is known to be NP-hard, as well as hard to approximate~\cite{storer1982data, CharikarLLPPSS05} (by some constant factor). The best known approximation ratio for the grammar-based compression of an input of length $n$ is $O(\log{n})$~\cite{rytter2003application, jez2014really}. 
The approximation ratio of Re-Pair is at most $O((n/\log{n})^{2/3})$ and at least $\Omega(\log{n}/\log\log{n})$~\cite{CharikarLLPPSS05, bannai2020smallest, hucke2017approximationratiorepair}. Note that these results relate the size of the obtained grammar to that of the minimal grammar, 
where the latter can be of a more general form than what Re-Pair (or BPE) can produce.

Navarro and Russo~\cite{navarro2008re} show a different kind of theoretical guarantee for Re-Pair, namely that its cost approximates the order-$t$ entropy of the input, for a certain range of $t$ and alphabet size. 
Furuya et al.~\cite{MRRepair} bound the gap between runs of Re-Pair with different tie-breaking. These results relate to a different measure of optimality, compressed length instead of reduction, they are thus not directly comparable to ours (we discuss these criteria in detail later); nevertheless, a construction from~\cite{MRRepair} turns out to be useful in our context.

Closest to our work is the recent paper of Zouhar et al.~\cite{zouhar2023formal}, which initiated the formal study of BPE and optimal pair encoding that we also largely follow in this paper, apart from small notational differences and a more general problem formulation.   
Using the theory of submodular functions, they relate the approximation ratio of BPE to a certain parameter (total backward curvature) that depends on the unknown optimum. Zouhar et al.\ observe an empirical bound on this quantity, however, without giving any worst-case guarantees.

\subparagraph{Problem definition.} 

We consider strings (sequences of symbols) over some alphabet $\Sigma$ and denote concatenation of strings $a$ and $b$ by $a \cdot b$, omitting the $\cdot$ when clear from context.  
We denote the \emph{length} of a string $s$ by $|s|$, the $i$-th symbol of $s$ by $s[i]$, and the substring $s[i] \cdot s[{i+1}] \cdots  s[j]$ by $s[i:j]$.  
We thus have $s = s[{1:|s|}]$.
 
A \emph{replacement rule} is a function  $\replace_{x \rightarrow y}$ that transforms a string $s$ by replacing \emph{all} occurrences of the string $x$ in $s$ with the string $y$. Formally, $\replace_{x \rightarrow y}(s) = s$ if $s$ does not contain $x$, and otherwise $\replace_{x \rightarrow y}(s) = s[{1:i-1}] \cdot y \cdot \replace_{x \rightarrow y}(s[{i+|x|:|s|}])$, where $i$ is the smallest index for which $s[{i:i+|x|-1}] = x$.

A sequence of replacement rules $\cR = (\cR_1, \dots, \cR_k)$ with $\cR_i = \replace_{a_i b_i\rightarrow c_i}$, where $a_i, b_i, c_i$ are symbols, is called a \emph{merge sequence} of length $k$. Denoting $s' = \left(\cR_k\circ \cdots \circ \cR_1\right)(s)$, where $\circ$ is the function composition, we refer to $|s'|$ as the \emph{compressed length}, and $|s|-|s'|$ as the \emph{utility} of $\cR$ for $s$. In words, $s'$ is obtained from $s$ by applying the sequence of replacement rules $\cR_1, \dots, \cR_k$.
We refer to the $i$-th step, i.e., the application of $\cR_i$ as the $i$-th \emph{merge}. We sometimes use the term \emph{full merge} to emphasize that no copies of $a_i b_i$ remain after the operation.  

Given the resulting encoded (compressed) string $s'$, we can recover $s$ by applying the sequence of reverse transformations $\cR' = (\cR'_k, \dots, \cR'_1)$ to $s'$, with $\cR'_i = \replace_{c_i \rightarrow a_i b_i}$. 
The symbols $c_i$, for all $i \in [k]$, are assumed to be new, i.e., not appearing in $s$ or in $\left(\cR_j\circ \cdots \circ \cR_1\right)(s)$ for $j < i$. Indeed, if $c_i$ already appears in the string, then the replacement $a_i b_i \rightarrow c_i$ may not be unambiguously reversible. 

\medskip

We can now formulate our main optimization problems. Given a string $s$ and an integer $k>0$, find a merge sequence $\cR$ of length $k$, of maximal utility for $s$ (or equivalently, of minimal compressed length). 
We denote this optimal utility as $\OPT^m(s,k)$, and call the task of computing it the \textbf{optimal merge sequence (OMS) problem}.\footnote{Apart from slightly different notation that is more convenient for our arguments, the OMS problem is identical to the problem defined by Zouhar et al.~\cite{zouhar2023formal}.} Note that maximizing compression utility and minimizing compressed length are equivalent for exact computation, but not necessarily for approximability.

\medskip

We also define a more general optimization problem where we do not require to replace \emph{every occurrence} of a pair. 
Formally, a \emph{partial replacement rule} $\cR_i^* = \replace^*_{a_i b_i \rightarrow c_i}$ can be any function that satisfies $\replace_{c_i \rightarrow a_i b_i}(\replace^*_{a_i b_i \rightarrow c_i}(s)) = s$ for all $s$. In words, $\cR_i^*$ replaces \emph{some} occurrences of $a_i b_i$ with $c_i$. A sequence of partial replacement functions $\cR^* = (\cR_1^*, \dots, \cR_k^*)$ is a \emph{partial merge sequence}. Denoting $s' = \left(\cR_k^* \circ \cdots \circ \cR_1^*\right)(s)$, we define utility and compressed length of $\cR^*$ analogously to merge sequences. Clearly, $s$ can be recovered from $s'$ identically to the case of merge sequences. 

The \textbf{optimal pair encoding (OPE) problem} asks, given a string $s$ and an integer $k > 0$, to find a partial merge sequence $\cR^*$ of length $k$, of maximal utility for $s$. We denote this optimal utility as $\OPT(s,k)$.  

While the OMS problem is perhaps more natural, OPE is more general, and as shown in Figure~\ref{fig2}, it can indeed be stronger (i.e., it is sometimes worth not merging every occurrence of a pair). Most of our results in this paper apply to both problems.  

\begin{figure}
\captionsetup{width=.92\linewidth}
\begin{align*}
\quad\quad  \mathtt{abcd\,|\,bc\,|\,bcda\,|\,cd\,|\,cdab\,|\,da\,|\,dabc\,|\,ab} & \mathtt{~\rightarrow~ XY\,|\,Z\,|\,ZT\,|\,Y\,|\,YX\,|\,T\,|\,TZ\,|\,X}\\
\quad\quad  \mathtt{abcd\,|\,bc\,|\,bcda\,|\,cd\,|\,cdab\,|\,da\,|\,dabc\,|\,ab} & \mathtt{ ~\rightarrow~ XZ\,|\,Y\,|\,YT\,|\,Z\,|\,ZX\,|\,T\,|\,dXc\,|\,X}
\end{align*}

\caption{Input $s = \mathtt{abcd\,|\,bc\,|\,bcda\,|\,cd\,|\,cdab\,|\,da\,|\,dabc\,|\,ab}$, where $|$ denotes a distinct symbol for each occurrence. An optimal OPE encoding of $s$ (above) with utility $\OPT(s,4) = 12$. An optimal OMS encoding of $s$ (below) with utility $\OPT^{m}(s,4)=11$. The OMS solution is obtained via the merge sequence $(\mathtt{ab \rightarrow X,~ bc \rightarrow Y,~ cd \rightarrow Z,~ da \rightarrow T})$. \label{fig2}}
\end{figure}

\subparagraph{Byte-pair encoding (BPE).} BPE solves both problems (OPE and OMS) as follows. Starting with the input string $s$, it performs $k$ \emph{locally optimal} full merge steps, always choosing a pair whose replacement maximizes compression utility. 
Formally, for input $(s,k)$, it outputs $\cR = (\cR_1, \dots, \cR_k)$, where $\cR_i = \replace_{a_i b_i \rightarrow c_i}$. Denoting $s^{(0)} = s$, and $s^{(i)} = \cR_i(s^{(i-1)})$ for $i \in [k]$, each $c_i$ is a new symbol, i.e., not occurring in $s^{(j)}$ with $j<i$, and for $i=1,\dots,k$, the pair $a_i b_i$ is chosen so that $|\cR_i(s^{(i-1)})|$ is minimal.  

With careful data structuring, identifying $a_i b_i$ and performing $\cR_i$ can be done in linear total time over all $k$ merge steps, e.g., see~\cite{shibata1999byte}.  
In this paper, we ignore such implementation details and focus on the total utility of BPE, i.e., $|s| - |s'|$, where $s' = s^{(k)}$. We denote this quantity as $\bpe(s,k)$. Observe that $\bpe(s,k) \leq \OPT^m(s,k) \leq \OPT(s,k)$ holds.

\medskip

We remark that a number of choices allow for small variation in the definition of BPE (and partly of $\OPT^m$): (1) when choosing a pair to replace, in case of a tie in utility, we pick the pair that appears first; (2) the utility of a chosen pair may be smaller than its number of occurrences in case of overlaps (e.g., the pair $\mathtt{aa}$ appears twice in $\mathtt{aaa}$, but only one of its copies can be replaced) -- one could pick by number of occurrences, instead of utility; and (3) in case of such overlapping pairs, we do the replacements left-to-right, e.g., $\mathtt{aaa \rightarrow Xa}$, whereas $\mathtt{aaa \rightarrow aX}$ would also be a valid choice. 

Overall, the effect of these design decisions appears negligible. Our results hold unchanged regardless of the tie-breaking strategy for (1). As for (2) and (3), the implementation we chose appears better motivated than the alternatives, but our results can easily be adapted to the other variants; see also~\cite{larsson2000off, zouhar2023formal} for discussion. 

\subparagraph{Our results.} 

As defined, OPE and OMS are  natural string compression problems (maximizing compression utility), and BPE is a straightforward greedy heuristic for both. Surprisingly, no worst-case guarantee is known for BPE or for any other algorithm solving OPE or OMS. 

Zouhar et al.~\cite{zouhar2023formal} raised the question of whether the  OMS decision problem is NP-hard. 
Our first result, shown in \S\,\ref{sec2}, answers this question in a stronger form. We show that both OMS and OPE are in fact APX-complete, ruling out the existence of a \emph{polynomial time approximation scheme} (PTAS), unless P=NP. 
\begin{theorem}\label{thm1}
OPE and OMS are APX-complete.
\end{theorem}

The fact that the number $k$ of merge-steps is part of the input is crucial; for fixed values of $k$ a polynomial-time exact algorithm is easy to derive. The APX-hardness also holds for the problem of minimizing compressed length, as well as for some other variants of the problem, as discussed in \S\,\ref{sec2}.

As for BPE, we analyze its approximation ratio for compression utility, showing in \S\,\ref{sec3}:

\begin{theorem}\label{thm2}
BPE approximates OPE with a ratio of $\alpha$, where $0.3333 \leq \alpha \leq 0.4324$. 
\end{theorem}

Closing this gap is an intriguing open question. The same bounds also hold for the approximation ratio of BPE for OMS. Note that the lower bound implies that $\OPT$ and $\OPT^{m}$ are within a constant factor of each other. Unlike the APX-hardness result, the  guarantee of Theorem~\ref{thm2} does not transfer to the dual optimization problem of minimizing compressed length. In particular, we show in \S\,\ref{sec4} that BPE cannot achieve a constant approximation for this measure. 

\begin{theorem}\label{thm3}
The BPE approximation ratio for OPE or OMS compression length is $\Omega(n)$. 
\end{theorem}

While our main focus is on the BPE algorithm, we find the OPE optimization problem of independent interest. We give in \S\,\ref{sec5} a simple algorithm we call EvenOdd, that achieves an approximation ratio of $0.5$. We stress that despite this guarantee, on typical inputs BPE appears to perform better than EvenOdd and the latter should be seen as a proof of concept. 

\begin{theorem}\label{thm4}
EvenOdd is a linear time $0.5$-approximation for OPE. 
\end{theorem}

The following four sections are dedicated to the proofs of Theorems~\ref{thm1}--\ref{thm4}. In \S\,\ref{sec7} we conclude with some open questions.

\section{Hardness of Approximation}\label{sec2}

In this section we show that the optimal merge sequence (OMS) problem is APX-complete, proving part of Theorem~\ref{thm1}. The APX-completeness of OPE is based on similar ideas and is completed in Appendix~\ref{appA}. As we show in \S\,\ref{sec3}, both OPE and OMS are in APX (i.e., admit a constant-factor approximation), it just remains therefore to show APX-hardness. We do so via an $L$-reduction (linear reduction, see~\cite{PapadimitriouY91}) from \emph{maximum cut} in cubic graphs. 

In this problem, given an undirected graph $G=(V, E)$, 
the task is to partition the vertex set $V$ into two parts, such as to maximize the number of edges 
between the two parts. 
Maximum cut is NP-complete~\cite{Karp72}, and its optimization version is APX-complete~\cite{PapadimitriouY91} even 
when restricted to input graphs in which every vertex has degree exactly three~\cite{AlimontiK00}. 
Our reduction is from such (cubic, undirected, unweighted) instances $G = (V,E)$ of the maximum cut problem. 

Assume $V = \{v_1, \dots, v_n\}$ and associate a symbol $\ell_i$ to each vertex $v_i \in V$. We construct a string $s$ by concatenating, for each edge $\{v_i, v_j\} \in E$ in arbitrary order, a string $s_{ij}$, with $i<j$. Then, we append for each vertex $v_i \in V$ four copies of a string $s_i$. Finally, we append $20n$ copies of a padding string $s_0$. These strings are defined as follows, with product signs indicating concatenation:
\begin{align*}
\quad\quad\quad s_{ij} & ~=~ \#\ell_i\#\#\ell_j\#\,|\,\#\ell_j\#\#\ell_i\#\,|, \\ 
\quad\quad\quad s_i & ~=~ \ell_i\#\ell_i\,|, \\
\quad\quad\quad s_0 & ~=~ \#\#\,|, \\
\quad\quad\quad s & ~=~ \prod_{\substack{\{v_i,v_j\} \in E\\i<j}}{s_{ij}} \cdot \prod_{v_i \in V}{s_i s_i s_i s_i} \cdot \prod_{t=1}^{20n}{s_0}.
\end{align*}

Here, 
$\#$ is a single fixed symbol, and $|$ denotes a distinct symbol for each occurrence.  
Setting $k = n+1$, we take $\OPT^{m}(s,k)$ to be the optimal utility of the resulting OMS instance.  
The construction clearly takes polynomial time. 

\medskip

We first show that $G$ has a cut of size at least $c$ if and only if $\OPT^{m}(s,k) \geq 30n + c$, i.e., if there is a merge sequence of length $k$, of at least the given utility for $s$. This already implies the NP-hardness of the problem and sets the stage for proving hardness of approximation. 

\medskip

$(\Rightarrow)$ Consider a cut of $G$, i.e., a partitioning $(S,V \setminus S)$ of $V$ with $c$ edges across the cut. 
We construct a merge sequence of length $k$ (recall that a \emph{merge} replaces every occurrence of a pair with a new symbol). We first merge $\# \ell_i$ for every $v_i \in S$, and $\ell_i \#$ for every $v_i \in V \setminus S$. Each merge step achieves utility $10$ ($6$ from the $s_{ij}$ corresponding to the three edges incident to $v_i$, and $4$ from the $s_i$ strings), for a total of $10n$. Then, we merge $\#\#$, which achieves utility $c+20n$. 
Indeed, after the first $n$ merges, one of the two $\#\#$ pairs remains in a string $s_{ij}$ exactly if the corresponding edge $\{v_i,v_j\}$ is in the cut, and otherwise no such pair remains in $s_{ij}$. The term $20n$ is from the padding string $s_0$. This adds up to a utility of $30n + c$ as claimed.

\medskip

$(\Leftarrow)$ Consider a merge sequence $\cR$ of length $k$, 
of utility $30n + c$ for $s$. We call the merge sequence \emph{well-formed}, if it includes merging the pair $\#\#$, and  for each $v_i \in V$ exactly one of $\ell_i\#$ and $\#\ell_i$; moreover, the merge $\#\#$ appears last. We show that it is sufficient to consider well-formed merge sequences.

\begin{lemma}\label{claim_wf}
If $\cR$ is not well-formed, then we can find (in polynomial time) a well-formed merge sequence $\cR'$ of the same length whose utility on $s$ is at least that of $\cR$.
\end{lemma} 

Before proving the lemma, we state a technical claim that will be useful. 

\begin{claim}\label{lem_removal}
Let $m_1, \dots, m_r$ be a merge sequence for an input string $s$, and
suppose every rule $m_{k+1}, \dots, m_r$ replaces only original symbols of $s$. If $m_k$ has utility $u_k$, then deleting
$m_k$ reduces the utility of the sequence by at most $u_k$. 
For partial merge sequences, the same assertion holds if the later merges retain the same replaced occurrences; their utility is then unchanged.
\end{claim}
\begin{proof}
Merges before $m_k$ are unaffected, so it suffices to show that the later merges perform at least as many replacements without $m_k$ as with it. Compare the original execution with the execution without $m_k$, and process all
later replacements in the order of the sequence, and for a given merge, from left to right. Skipping $m_k$ can restore original-symbol occurrences and adjacencies, but cannot destroy them, since the output of $m_k$ was a new symbol. 

Every occurrence replaced only in the
original execution must be blocked in the shortened execution by an earlier adjacent occurrence
replaced only there; charge the lost occurrence to this blocking occurrence. No shortened-only
occurrence receives two charges: if it contains a restored symbol, only its other endpoint can occur
in the original execution, while if both endpoints occur there, two charges would imply that both
were free when this occurrence was considered in the original execution, so it would also have been
replaced there. Hence the shortened suffix performs at least as many replacements.

For partial merges, retain the same later replacements, identified by their boundaries in $s$.
Skipping $m_k$ does not destroy any such boundary, so all remain feasible.
\end{proof}
\medskip

We now proceed with the lemma on well-formed sequences.

\begin{proof}[Proof of Lemma~\ref{claim_wf}]
Any new symbol introduced in a merge encodes a substring of $s$. Since pairs other than $\#\#$ occur at most $10$ times in $s$, a string different from $\#\#$ can also occur at most this many times. Thus, any merge other than $\#\#$ can have utility at most $10$. 

Recall that \emph{any} well-formed merge sequence has utility at least $30n$, by the earlier correspondence to an arbitrary cut, possibly of size zero. 
So, if $\cR$ does not merge $\#\#$, then its total utility is at most $10(n+1) < 30n$ and we can replace $\cR$ by an arbitrary well-formed merge sequence of larger utility. Assume therefore that $\#\#$ is merged by $\cR$. 

Suppose $\cR$ is not well-formed, and that there exist $t$ indices (for some $t>0$) $i_1, \dots, i_t \in [n]$ such that none of the two pairs $\ell_{i_j}\#$, $\#\ell_{i_j}$ 
are merged by $\cR$, for all $j \in [t]$. Instead, there must be exactly $t$ ``bad'' merges, which can be of the types listed below. 
We show that removing the $t$ bad pairs from $\cR$ and adding the missing pairs $\ell_{i_j}\#$, for all $j \in [t]$, we obtain a merge sequence $\cR'$ of larger utility on $s$.

Since $\cR'$ contains the set of merges of a well-formed sequence, each newly added pair contributes a utility of at least $7$. 
This is because the pair appears once in each of the three relevant strings $s_{ij}$, in copies that do not overlap with $\#\#$, and thus contribute to the utility. (Note that we did not yet assume  anything about when the $\#\#$-merge happens.) An additional utility of $4$ comes from the four copies of $s_{i}$. As we add the pairs to the end of the merge sequence, they do not reduce the utilities of other merges. 

To finish the proof, we argue that the removal of each ``bad'' merge from $\cR$ decreases utility by at most $6$; as the remaining merges are unaffected, we obtain an overall increase in utility of at least one for each bad merge. Bad merges that deviate from a well-formed merge sequence can be of the following types:

\begin{enumerate}
\item A merge involving the symbol $|\,$: this yields a utility of one, because every separator is unique. 
\item A merge involving a new symbol (not part of the original $s$): such a merge corresponds to a substring of $s$ of length at least three. Of all such substrings, those involving $|$ occur at most once, $\#\ell_i\#$ occurs $6$ times, $\#\#\ell_i$ and $\ell_i\#\#$ occur $3$ times, and $\ell_i\#\ell_i$ occurs $4$ times, for all $i$; there are no further cases (longer substrings also cannot occur more than $6$ times). Thus any such merge can yield utility of at most $6$.
\item A merge of $\ell_i\#$ after $\#\ell_i$ has already been merged, or a merge of $\#\ell_i$ after $\ell_i\#$ has already been merged: assuming that bad merges of the first two types have been removed already, this yields utility at most $6$. To see this, observe that for all such pairs, initially there are $6$ copies in $s_{ij}$ strings and $4$ copies in $s_i$ strings, and the $4$ copies in $s_i$ have been destroyed by merging the other pair.  
\item Repeated or vacuous merges (pairs not appearing in the string) have utility $0$, we can thus assume they do not occur. 
\end{enumerate}

Delete all bad merges of types $1$ and $2$ from right to left. At every deletion, all later remaining
merges use only original symbols. 
By Claim~\ref{lem_removal}, deleting the current
merge decreases total utility by at most its current utility, hence by at most $6$.
Afterwards delete the remaining bad merges, again from right to left. Now every remaining merge
uses original symbols, so Claim~\ref{lem_removal} applies again. A merge of type $3$ has utility at most six by the
preceding argument, thus deleting all $t$ bad merges loses at
most $6t$, while appending the $t$ missing merges gains at least $7t$.
Finally, moving $\#\#$ last cannot decrease utility. In every edge half $\#\ell_x\#\#\ell_y\#$, the two selected
endpoint pairs are disjoint and can overlap only the central $\#\#$; every central replacement lost
by postponing $\#\#$ is offset by at least one newly enabled endpoint replacement. The
padding and vertex gadgets are unaffected. We obtain the required well-formed sequence.
\end{proof}

Suppose that $\cR$ is well-formed (otherwise, transform it according to Lemma~\ref{claim_wf}). From $\cR$, compute a cut $(S,V\setminus S)$ of $G$ by letting $S$ be the set of vertices $v_i \in V$ where merge $\#\ell_i$ is in $\cR$. 
Consider the string $s_{ij}$ for an edge $\{v_i, v_j\}$. We can achieve a total utility of $5$ for this string if and only if one of $v_i$ and $v_j$ was placed in $S$ and the other was not, i.e., if the edge contributes to the constructed cut. 
Otherwise, both $\#\ell_i$ and $\#\ell_j$, or both $\ell_i\#$ and $\ell_j\#$ are merged, and a utility of at most $4$ is achievable for $s_{ij}$, as both $\#\#$ pairs overlap with the two other pairs. 

Of the total utility of at least $30n+c$ we obtain $20n$ in the $s_0$-strings, $4n$ in the $s_i$-strings, with $6n + c$ remaining for the $s_{ij}$-strings corresponding to the $3n/2$ edges of $G$. By the above calculations these can contribute exactly $4 \cdot 3n/2 = 6n$ plus the size of the constructed cut, which must therefore be at least $c$. This concludes the reduction for NP-hardness.

\medskip

We next turn the argument into an L-reduction~\cite{PapadimitriouY91}, thus showing the APX-hardness of OMS. Let $\OPT^c$ denote the maximum cut size for graph $G$, and let $\OPT^m = \OPT^{m}(s,k)$ be the maximum utility for the OMS instance resulting from the above reduction. It remains to show that there exist positive constants $\alpha$ and $\beta$ (independent of the input $G$) so that, for all sufficiently large $G$:
\begin{enumerate}
\item $\OPT^{m} \leq \alpha \cdot \OPT^c$, and
\item for every feasible solution of the OMS instance $(s,k)$ of utility $u$ we can obtain, in polynomial time, a cut of $G$ of size $c$, with $|\OPT^c - c\,| \leq \beta \cdot |\OPT^m - u\,|$.
\end{enumerate}

Part 1.~is immediate: since $G$ is cubic with $n$ vertices, it has $3n/2$ edges, and we can add up the lengths of the different parts as $|s| = 3n/2 \cdot 14+ n \cdot 4 \cdot 4 + 20n \cdot 3 = 97  n$, which upper bounds $\OPT^m$; on the other hand, from the well-known lower bound on the maximum cut (half the number of edges) we have $\OPT^c \geq 3n/4$, implying the inequality for $\alpha = 130$.  

For Part 2., consider a merge sequence $\cR$ of utility $u$ for $(s,k)$. If $\cR$ is not well-formed, we use the transformation of Lemma~\ref{claim_wf} to obtain a well-formed merge sequence $\cR'$ for $(s,k)$ of utility $u' \ge u$. Otherwise let $u' = u$ and $\cR' = \cR$.
Now, as $\cR'$ is well-formed, we interpret it as a cut of size $c$ in $G$, where, by the earlier reduction, $c = u'-30n$. Moreover, by applying this argument to the optimal OMS solution, we obtain a cut of size $\OPT^m - 30n$.  As we have shown earlier that $\OPT^m \geq 30n + c'$ for every cut $c'$, our cut must be maximal and $\OPT^c = \OPT^m - 30n$.  
We thus have $|\OPT^c - c\,| = |\OPT^m - u'\,| \leq |\OPT^m - u\,|$, proving the inequality for $\beta = 1$.

This completes the APX-hardness of OMS. In Appendix~\ref{appA} we extend the argument to show the APX-hardness of the OPE problem. 

\medskip

Two small remarks are in order. First, we note that the APX-hardness also applies to the compressed length minimization version of OMS (and OPE); in the above hard instance the input length is $|s| = 97n$, and $\OPT^m(s,k) \geq 30n$.  
Thus,  
a $(1+\varepsilon)$-approximation for $|s'| = |s| - \OPT^m(s,k)$ would also yield a $(1-\varepsilon')$-approximation for $\OPT^m$, with $\varepsilon' = \left((97-30)/30\right)\varepsilon \leq 2.24\,\varepsilon$. 

Second, we could restrict the OMS (and OPE) problems such as to only allow merging pairs of original input symbols, and not pairs involving new symbols. (Equivalently, all resulting tokens would be of length at most two; note that the BPE solution does not conform to this restriction.) Since in the above reduction optimal merges involve only input symbols, our APX-hardness results immediately transfer to these restricted variants. 

\section{Approximation ratio of Byte-Pair Encoding} \label{sec3}

In this section we study the approximation ratio of BPE for the OPE problem, showing constant upper and lower bounds for this ratio, proving Theorem~\ref{thm2}. 

\subsection{Upper bounds} 
We start with two simpler constructions that show the upper bounds of $5/8 = 0.625$, and $1/2$, respectively. 
The more involved construction that yields the bound stated in Theorem~\ref{thm2} is given in Appendix~\ref{appB}. 

\subparagraph{An upper bound of $5/8$.}
Consider the string $s = \mathtt{abaacaaba\,|\,aca}$, with a merge sequence of length $k=4$. Suppose BPE first performs the merge $\mathtt{aa \rightarrow X}$ (in an unfavorable tie-break), resulting in the string $\mathtt{abXcXba\,|\,aca}$. As now each pair occurs at most once, subsequent merges can only shorten the string by $1$, for a total utility of $5$. 
Consider now the alternative merge sequence $(\mathtt{ac \rightarrow X}$,~ $\mathtt{Xa \rightarrow Y}$,~ $\mathtt{ab \rightarrow Z}$,~ \allowbreak $\mathtt{Za \rightarrow T})$ resulting in $\mathtt{TYT\,|\,Y}$, with utility $8$, and a ratio of $5/8 = 0.625$.

To enforce that BPE merges $\mathtt{aa}$ first, we can concatenate $t$ copies of $s \cdot \#$, and a final $\mathtt{aa}$. The utility of BPE is now (without relying on tie-breaking) $5t+1$, whereas the alternative merge sequence yields $8t$, for a ratio arbitrarily close to $0.625$. Note that while we state the result for OPE, it also (more strongly) applies to OMS (i.e., for the ratio $\bpe(s,k)/\OPT^m(s,k)$), as the alternative encoding above is given by a full merge sequence. 

\subparagraph{An upper bound of $1/2$.}

We next show that the approximation ratio of BPE is at most $1/2$.
The construction will use $t + t^2$ symbols, denoted as $\mathtt{a_i}$, resp., $\mathtt{a^{j}_i}$, for $i,j \in [t]$, as well as separator symbols $\,|\,$ that stand for a unique symbol in each occurrence. 
Let us denote $$s = \prod_{\ell = 1}^t \prod_{i,j \in [t]} { {\mathtt{a^{\ell}_i a_i a_j a^{\ell}_j} \,|}},$$ and consider a merge sequence of length $k = 2t^2$.

In words, the input string consists of $t$ copies of each pair $\mathtt{a_i a_j}$, including the case $i=j$, flanked from left and right by $\mathtt{a^{\ell}_i}$, resp., $\mathtt{a^{\ell}_j}$, with $\ell$ taking all $t$ possible values.
For example, for $t=2$, we have $$s = \mathtt{a^{1}_1 a_1 a_1 a^{1}_1} \,|\, \mathtt{a^{1}_1 a_1 a_2 a^{1}_2} \,|\, \mathtt{a^{1}_2 a_2 a_1 a^{1}_1} \,|\, \mathtt{a^{1}_2 a_2 a_2 a^{1}_2} \,|\,\mathtt{a^{2}_1 a_1 a_1 a^{2}_1} \,|\, \mathtt{a^{2}_1 a_1 a_2 a^{2}_2} \,|\, \mathtt{a^{2}_2 a_2 a_1 a^{2}_1} \,|\, \mathtt{a^{2}_2 a_2 a_2 a^{2}_2} \,|\,.$$

A possible (in fact, optimal) merge sequence can merge $\mathtt{a^{\ell}_i a_i \rightarrow A^{\ell}_i}$ and $\mathtt{a_i a^{\ell}_i \rightarrow B^{\ell}_i}$ for each $i, \ell\in [t]$, where $\mathtt{A^{\ell}_i}$ and $\mathtt{B^{\ell}_i}$ are new symbols. 
As each merged pair appears $t$ times, we obtain a total utility (over $2t^2$ merges) of $2t^3$, implying $\OPT^m(s,k) \geq 2t^3$.

With an unfavorable tie-breaking BPE will perform, as the first $t^2$ merges, those of the form $\mathtt{a_i a_j \rightarrow A_{ij}}$, each of utility $t$. After this, each pair occurs at most once, so the remaining $t^2$ budget is spent on utility-$1$ merges. The total utility is thus $t^3 + t^2$, for a ratio $\bpe(s,k)/\OPT^m(s,k) \leq 1/2 + 1/{2t}$, arbitrarily close to $1/2$ for large $t$. 
We can easily enforce the unfavorable tie-breaking, as before, by concatenating copies of $s$, together with a tie-breaker gadget, still yielding a ratio arbitrarily close to $1/2$.

Again, the result applies to BPE, with respect to both OPE and OMS. 

We note that the construction also relates to Re-Pair, showing that the difference between the resulting grammar sizes with different tie-breaking strategies can be \emph{linear} in the input string length $|s|$. This settles an open question about Re-Pair and improves the previous lower bound of~\cite[Thm.\,3]{furuya2019mrrepairgrammarcompressionbased} which had the form $\Omega(\sqrt{|s|})$.

\subsection{Lower bound}

We now show our main positive result for BPE, a lower bound of $1/3$ on its approximation ratio. We show this bound (more strongly) with respect to OPE. Consider an optimal (possibly partial) merge sequence $\cR^* = (\cR_1^*,\dots,\cR_k^*)$, where $\cR_i^* = \replace^*_{a_i b_i \rightarrow c_i}$ for $i \in [k]$. Let $s'$ be the string resulting from applying $\cR^*$ to input $s$. Recall that $\OPT(s,k) = |s| - |s'|$. 

In order to relate $\bpe(s,k)$ and $\OPT(s,k)$, we first introduce a simpler quantity that upper bounds both. 

For a string $s$ of length $n$, a set of $t$ indices $\{i_1, \dots, i_t\} \subseteq [n-1]$ is a \emph{pair packing}, if the pairs starting at the given indices of $s$ do not overlap and are all equal as strings. Formally, $|i_{j}-i_k| \geq 2$ for $j,k \in [t]$, and $s[i_1: i_1+1] = \cdots = s[i_t: i_t+1]$. 
A \emph{$k$-packing} is the union of $k$ disjoint pair packings, i.e., no index appears in more than one of the pair packings. We refer to a largest possible $k$-packing as the \emph{optimal $k$-packing} and denote its size by $P_k(s)$. We remark that $P_k(s)$ can be computed in polynomial time, but this is not needed in our current study. 

\begin{lemma}\label{lem_freq}
For $s'$ obtained from $s$ by applying an arbitrary (partial) merge sequence $\cR^*$ of length $k$, we have $|s|-|s'| \leq P_k(s)$.
\end{lemma}

\begin{proof}

We create a $k$-packing from $s$ and $\cR^*$, of size $|s| - |s'|$. To do this, it will be useful to map every occurrence of a symbol during the application of $\cR^*$, to the substring of $s$ which is encoded by the symbol. For a string $c$ let $\decode(c,0)=c$, and for $i \in [k]$ let $\decode(c,i) = \decode(\replace_{c_i \rightarrow a_i b_i}(c),i-1)$. 
Observe that every occurrence of $c_i$ after the application of $\cR_i^*$ encodes a unique copy of $\decode(c_i, i)$ in $s$, these copies do not overlap, and $\decode(c_i, i) = \decode(a_i, i-1) \cdot \decode(b_i, i-1)$. 

We \emph{charge} each replacement $a_i b_i \rightarrow c_i$, i.e., each unit of the compression utility, in the order of their application in $\cR^*$, to a certain index of $s$. More precisely, charge the replacement $a_i b_i \rightarrow c_i$ to the index in $s$ of the last symbol of $\decode(a_i, i-1)$ in the location encoded by the current copy of $a_i$. Intuitively, we charge the merging of a pair to the location in $s$ where the ``gluing'' happens. 

Observe that this index cannot be charged again in any step, since after replacement $a_i b_i \rightarrow c_i$, it points to a symbol ``inside'' $\decode(c_i, i)$, that is not the last one, and future merges cannot cause it to become a last symbol.  The indices charged due to $\cR_i^*$ are at pairwise distance at least $2$, for all $i \in [k]$ (as otherwise two copies of $\decode(c_i,i)$ would overlap in $s$), and the pairs starting at the charged indices are equal as strings, as all consist of the last symbol of $\decode(a_i, i-1)$ and first symbol of $\decode(b_i, i-1)$. It follows that the set of indices charged during the entire process is a $k$-packing, thus, the total utility $|s| - |s'|$ is at most the size $P_k(s)$ of the optimal $k$-packing.  See Figure~\ref{fig4} for illustration.
\end{proof}

\begin{figure} 
\centering\includegraphics[scale=0.32]{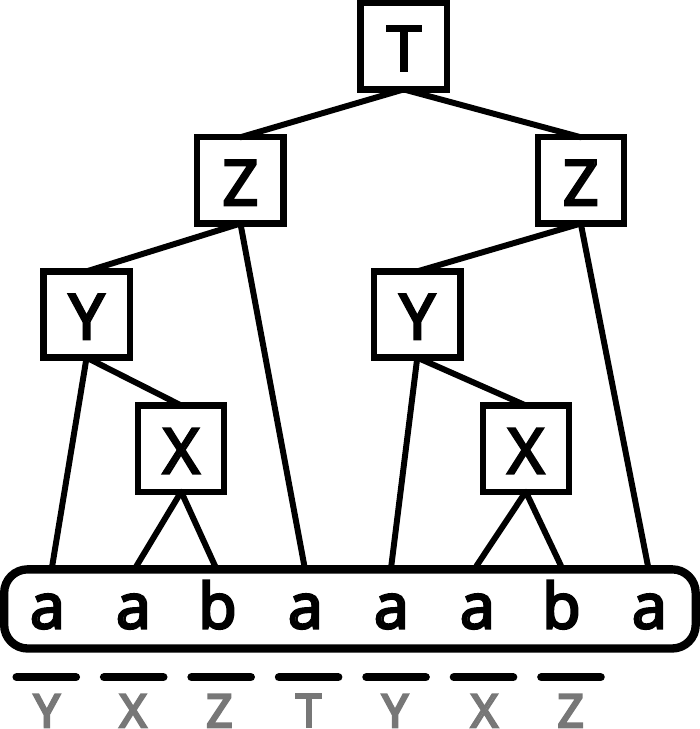}
\captionsetup{width=.92\linewidth}

\caption{Sequence of merges $\mathtt{aabaaaba \rightarrow aXaaXa \rightarrow YaYa \rightarrow ZZ \rightarrow T}$ shown as a tree (above) and corresponding $k$-packing  (below).\label{fig4} }
\end{figure}

Observe that as a corollary of Lemma~\ref{lem_freq} we have $\OPT(s,k) \leq P_k(s)$.

\medskip

In the remainder of the section we \emph{lower bound} the BPE utility in terms of the quantity we defined, which, together with Lemma~\ref{lem_freq} will immediately imply the claimed lower bound on the approximation ratio of BPE. 

\begin{lemma}
For all strings $s$ and $k\geq 0$ we have $\bpe(s,k) \geq P_k(s)/3$.
\end{lemma}

\begin{proof}

As before, $s^{(0)} = s$ and let $s^{(k)} = s'$ denote the encoding obtained by BPE, with $s^{(i)}$ resulting from $s^{(i-1)}$ (for $i \in [k])$ by a greedy merge step, as described in the definition of BPE.

Let $a_i b_i$ be the pair merged by BPE when going from $s^{(i-1)}$ to $s^{(i)}$, for $i \in [k]$, and let $\cB$ be a $k$-packing derived from the BPE merge sequence according to the process in Lemma~\ref{lem_freq}. Let $\cB_1 \cup \cdots \cup \cB_k$ be the partitioning of $\cB$ into the $k$ pair packings corresponding to the $k$ merge steps of BPE.
(When BPE does a replacement $a_i b_i \rightarrow c_i$ in $s^{(i-1)}$, then the corresponding entry in $\cB_i$ is the index of the last symbol of the substring in $s$ encoded by $a_i$.) Note that $|\cB| = \bpe(s,k)$.

Let $\cP$ be an optimal $k$-packing partitioned into $k$ pair packings $\cP_1 \cup \cdots \cup \cP_k$, where $\cP_i$ contains starting indices of the pair $x_i y_i$ in $s$ for $i \in [k]$. Note that $|\cP| = P_k(s)$. Also note that pairs $x_i y_i$, $x_j y_j$ from different pair packings are not necessarily distinct. For BPE, however, pairs $a_i b_i$, $a_j b_j$ from different steps are in fact, distinct. 

\medskip

We iterate $i$ from $1$ to $k$, relating for each $i$ the number of entries in $\cB_i$ and the number of entries in $\cP_i$. In each step we delete at most $3 |\cB_i|$ elements from $\cP$. The set $\cB$ is not modified during the process. In the end, $\cP$ will be empty, implying the claim. 

\medskip

We maintain two invariants:

$\mathbf{I_1:}$ After step $i$, the sets $\cP_j$ for $j \leq i$ are empty. 

$\mathbf{I_2:}$ After step $i$, no index $x \in \cP$ is a neighbor of an index $y \in \cB_j$, for $j \leq i$. \\
Two indices $x$ and $y$ are neighbors if $|x-y|=1$.

\medskip
Intuitively, $\mathbf{I_1}$ ensures that all elements of $\cP$ are accounted for, and $\mathbf{I_2}$ ensures that pairs $s[x:x+1]$ corresponding to $x \in \cP$ can still be merged by BPE, as the two symbols forming the pair did not take part in merges yet. 
Initially both invariants are trivially true. 

\medskip

Consider step $i$. For all elements $x \in \cB_i$, delete $x$ and its neighbors $x-1$ and $x+1$ from $\cP$. Note that we delete at most $3|\cB_i|$ elements in this way. Also note that this establishes $\mathbf{I_2}$ for the current step. It remains to verify $\mathbf{I_1}$. We distinguish two cases.

\medskip

Case 1. If $a_i b_i = x_\ell y_\ell$, for some $i \leq \ell \leq k$, then swap $x_i y_i$ with $x_\ell y_\ell$ and $\cP_i$ with $\cP_\ell$ (this is only necessary if $\ell \neq i$; note that the swap does not affect $\mathbf{I_1}$, $\mathbf{I_2}$). Now, if $\cP_i$ is nonempty (after the deletions and swap), then take any element $x \in \cP_i$. Note that BPE could have also merged $s[x:x+1]$ in the $i$-th step, in addition to the pairs indexed by $\cB_i$. This is because the pair equals $a_i b_i$, it does not neighbor any pair merged by BPE in the present or past rounds (by $\mathbf{I_2}$), and the pair itself was not merged before, as otherwise $x$ would have been deleted. This contradicts the definition of BPE, showing $\cP_i = \emptyset$, and hence, $\mathbf{I_1}$. 

\medskip

Case 2. If $a_i b_i \neq x_\ell y_\ell$ for $i \leq \ell \leq k$. In this case, $\cP \cap \cB_i = \emptyset$ even before the deletion.

Indeed, if $x \in \cP$ survived, then neither $x$ nor its neighbors lie in $\cB_1, \dots, \cB_{i-1}$, so positions $x$, $x+1$ still carry single original symbols in $s^{(i-1)}$; if additionally $x \in B_i$, then BPE's $i$-th merge glued exactly there, forcing $a_i=s[x]$, $b_i=s[x+1]$, i.e., $a_ib_i=x_\ell y_\ell$, which is Case $1$. Hence in Case $2$ the intersection is indeed empty.

If, before the deletion, $|\cP_i| > |\cB_i|$, then BPE could have merged, instead of the pairs $a_i b_i$ indexed by $\cB_i$, the pairs $x_i y_i$ indexed by $\cP_i$, contradicting the definition of BPE. This is because pairs $s[x:x+1]$ for $x \in \cP_i$ do not neighbor any pair merged by BPE in the present or past rounds (by $\mathbf{I_2}$), and were not merged before (as otherwise $x$ would have been deleted). Moreover, as $\cP_i$ forms a pair packing, the pairs indexed by $\cP_i$ are not neighboring each other and could thus have been merged in one round of BPE. 

Thus, $|\cP_i| \leq |\cB_i|$ and $\cP \cap \cB_i = \emptyset$ (before the deletion). As we cannot delete entries of $\cB_i$ from $\cP$, we instead use our remaining budget of $|\cB_i|$ deletions to delete the entries of $\cP_i$, establishing $\mathbf{I_1}$.  

\medskip

After the $k$-th step, $\mathbf{I_1}$ implies that $\cP$ is empty, proving the lemma. 
\end{proof}

\section{Lower bound for compressed length}
\label{sec4}

In this section we show that for the dual optimization problem of OPE of minimizing compressed length (instead of maximizing compression utility), the approximation ratio of BPE is \emph{unbounded}, and in fact, linear in the input size. This proves Theorem~\ref{thm3}.
    
The construction is essentially the same as the example used in \cite[Thm.~3]{MRRepair} to prove a different property of Re-Pair.

Consider the string 
    \begin{align*} \quad\quad\quad s = \prod_{i=1}^{t}x_i\mathtt{aa}y_i \cdot \prod_{i=1}^{t}|\,x_i\mathtt{a} \cdot \prod_{i=1}^{t}|\,\mathtt{a}y_i,
    \end{align*}
    where $t > 2$, each occurrence of $|$ is a distinct symbol, and $\mathtt{a}$, $x_i$, $y_i$ are symbols for $i \in [t]$. 
    
    If we merge every occurrence of $x_i\mathtt{a}$ and $\mathtt{a}y_i$, the resulting string has length $2t + 2t + 2t = 6t$ after $2t$ steps. So after $2t + (6t - 1) = 8t - 1$ total steps, we have reduced $s$ to a single symbol. 
    
    On the other hand, BPE will always start with $\mathtt{aa}$, since it occurs $t$ times, while every other pair occurs at most twice. After merging $\mathtt{aa}$, no pair occurs more than once in the string, so every subsequent step has a utility of one. After merging $\mathtt{aa}$, the length of the string is $3t+3t+3t = 9t$. So after $8t - 1$ total steps, the resulting string will be of length
        $9t - (8t-2) = t + 2$.
    
    Since $|s| = 10t$, it holds that $t+2 = \frac{|s|}{10} + 2 \in \Omega(|s|)$.

    Observe that since all merges were full, the approximation lower bound also holds for the OMS problem (in the compression length version).

\section{Approximating Optimal Pair Encoding}
\label{sec5}

In this section we give a simple algorithm that achieves a $0.5$-approximation for the OPE problem, proving Theorem~\ref{thm4}. 

For a pair $xy$, let $\freq_{xy}(s)$ denote the number of occurrences of $xy$ in $s$. Note that we allow overlaps, e.g., $\freq_{\mathtt{aa}}(\mathtt{aaa}) = 2$. Let $F_k(s)$ denote the total number of occurrences of the $k$ most frequent pairs in $s$ (breaking ties arbitrarily, as this does not affect the total value). Let $x_1 y_1, \dots, x_k y_k$ be such a set of most frequent pairs, and let $I=\{i_1, \dots, i_N\}$, where $N=F_k(s)$, denote the set of indices where they occur, i.e., $s[i_j:i_j+1] \in \{x_i y_i ~|~ i \in [k]\}$ for all $j \in [N]$. 

We can efficiently, i.e., in time $O(|s|)$, compute a set $I = \{i_1, \dots, i_N\}$ by enumerating all pairs occurring in $s$ and storing the indices of the $k$ most frequent of them. Assume that $i_1, \dots, i_N$ are sorted increasingly. We then sparsify $I$, by removing neighboring indices, through a single pass through $s$. More precisely, for each index of the form $i_{2\ell}$, remove $i_{2\ell}$ if $i_{2\ell} = i_{2\ell-1}+1$ or $i_{2\ell} = i_{2\ell+1}-1$. Let $I'$ denote the indices in $I$ remaining after this process. Observe that $I'$ has at least $N/2$ elements. As there are no neighboring indices in $I'$, the occurrences of pairs in $s$ indexed by $I'$ do not overlap, can therefore be merged simultaneously. Observe that only $k$ distinct pairs appear, so we obtain a valid OPE solution of utility $|I'|$. We call this the EvenOdd algorithm. The overall running time is clearly $O(|s|)$.

Since $|I'| \geq F_k(s)/2$, and clearly, $F_k(s) \geq P_k(s) \geq \OPT(s,k)$, the theorem follows. 

\medskip

\begin{figure}

\begin{align*}
\quad\quad\quad \mathtt{\underset{\uparrow}{a}b\underset{\uparrow}{c}d\,|\,\underset{\uparrow}{b}c\,|\,b\underset{\uparrow}{c}da\,|\,\underset{\uparrow}{c}d\,|\,\underset{\uparrow}{c}d\underset{\uparrow}{a}b\,|\,\underset{\uparrow}{d}a\,|\,d\underset{\uparrow}{a}bc\,|\,\underset{\uparrow}{a}b} & ~\rightarrow~  \mathtt{XY\,|\,Z\,|\,bYa\,|\,Y\,|\,YX\,|\,T\,|\,dXc\,|\,X}\\
\quad\quad\quad \mathtt{\underset{\uparrow}{a}\underset{{\bcancel{\uparrow}}}{b}\underset{\uparrow}{c}d\,|\,\underset{\uparrow}{b}c\,|\,\underset{\uparrow}{b}\underset{{\bcancel{\uparrow}}}{c}\underset{\uparrow}{d}a\,|\,\underset{\uparrow}{c}d\,|\,\underset{\uparrow}{c}\underset{{\bcancel{\uparrow}}}{d}\underset{\uparrow}{a}b\,|\,\underset{\uparrow}{d}a\,|\,\underset{\uparrow}{d}\underset{{\bcancel{\uparrow}}}{a}\underset{\uparrow}{b}c\,|\,\underset{\uparrow}{a}b} & ~\rightarrow~  \mathtt{XY\,|\,Z\,|\,ZT\,|\,Y\,|\,YX\,|\,T\,|\,TZ\,|\,X}
\end{align*}

\caption{Encoding $s = \mathtt{abcd\,|\,bc\,|\,bcda\,|\,cd\,|\,cdab\,|\,da\,|\,dabc\,|\,ab}$, where $|$ denotes a distinct symbol for each occurrence. The BPE encoding (above) with utility $\bpe(s,4) = 10$ via the merge sequence $(\mathtt{ab \rightarrow X,~ cd \rightarrow Y,~ bc \rightarrow Z,~ da \rightarrow T})$. Arrows show the $k$-packing solution derived from the BPE merge sequence. The EvenOdd encoding (below) with $k=4$ and utility $12$. Arrows show the indices of most frequent pairs, with crossed ones removed.    \label{fig3} }
\end{figure}

On some inputs, EvenOdd can have larger utility than BPE. Indeed, consider the earlier example $s = \mathtt{abcd\,|\,bc\,|\,bcda\,|\,cd\,|\,cdab\,|\,da\,|\,dabc\,|\,ab}$. With $k=4$ merge rounds, BPE achieves utility $10$ due to the unfavorable tie-breaking, but as seen earlier, the maximum utility of \emph{any} OMS algorithm (regardless of tie-breaking) is at most $11$. EvenOdd takes the most frequent $4$ pairs $\{\mathtt{ab,bc,cd,da}\}$ with $16$ occurrences in $s$. After sparsification, $12$ non-neighboring pairs remain and are merged, matching the OPE optimum utility of $12$ on this instance; see Figure~\ref{fig3}.

We remark that the EvenOdd algorithm only merges input symbols (and not newly introduced symbols), and thus, the decoded string of each symbol is of length at most two. Note that, instead of removing even-ranked indices, we could, in polynomial time, find the minimal set of indices whose removal makes the remainder non-overlapping (by a simple greedy strategy). 
In the worst case, however, no algorithm in this class can improve the approximation ratio $0.5$. 

Indeed, consider an input string $s$ consisting of $2n$ identical symbols, followed by $2(k-1)$ pairwise distinct symbols. Setting $k = \log_2{n}+1$, the optimal pair encoding of $s$ will collapse the first $2n$ symbols to one, with utility $2n-1$. An algorithm that can only merge input pairs can achieve a maximum utility of $n$ on the first part of the string (in one merge), and an additional $k-1$ on the last part, yielding an approximation ratio of $(n+\log_2{n})/(2n-1) = 0.5 + o(1)$.
In this example, the optimal encoding coincides with the one found by BPE, also showing a factor $2-o(1)$ gap between BPE and input-symbol-only algorithms, to the advantage of the former.

\section{Conclusion and open questions}\label{sec7}

In this paper we studied the complexity of optimal pair encoding: the task of compressing a string by replacing pairs with new symbols, such as to maximize the overall reduction in length. We showed that the problem is APX-complete and that BPE, a popular greedy heuristic used in state-of-the-art LLM pretraining, achieves a constant-factor approximation. 
Finding the best approximation ratio for the OPE or OMS problems, by BPE or by other polynomial-time algorithms, i.e., closing the gaps between our bounds in Theorem~\ref{thm2}, resp., Theorems~\ref{thm1} and \ref{thm4} are the central remaining questions. In particular: is there an efficient algorithm for OPE or OMS with approximation ratio strictly above $0.5$? Note that analyzing natural greedy strategies for some other string problems turned out to be very difficult. For instance, the famous \emph{greedy superstring conjecture}~\cite{Blum_string} concerns a very intuitive string merging process, with worst-case approximation ratio conjectured to be $2$ but only proven to be below $\approx 3.396$~\cite{SCS1,SCS2}. 

While the APX-hardness (Theorem~\ref{thm1}) extends to the compression length, the approximation guarantee (Theorem~\ref{thm2}) does not. The polynomial-time approximability of compression length (for OPE or OMS, by any algorithm) is left open. In fact, as we lack a constant-factor approximation for this problem (BPE does not achieve it, as shown in Theorem~\ref{thm3}), we can only claim its APX-hardness, not APX-completeness. 

Our hardness result (Theorem~\ref{thm1}) relies on an alphabet whose size increases with the input. 
In recent follow-up work, Kastreva et al.~\cite{kastreva2026} rule out the existence of a PTAS even with constant-size alphabets, when compressing \emph{collections} of strings. Whether the guarantees of BPE or other algorithms can be strengthened in the small-alphabet case remains an interesting question. 

\subparagraph{Note on using AI.} The original version of the paper, as appearing in the ESA 2026 proceedings, was written without AI use; it bounded the approximation ratio of BPE within $[0.333,0.625]$. The current updated version improves the upper bound to $(1-e^{-2})/2 \approx 0.4324$; this was obtained with the help of ChatGPT 5.6 in an interactive exploration of the problem, and was subsequently verified by the authors. This improvement was already presented at ESA 2026. The tool was also used to check existing proofs. A small technical inaccuracy in the hardness proof of \S\,\ref{sec2} was found during this process; the fix via Claim~\ref{lem_removal} makes an earlier hidden assumption explicit. Smaller improvements in the structure and presentation of the mentioned two proofs were also suggested by AI models (ChatGPT 5.6 Sol and Claude Opus 5.0); these were evaluated by the authors on a case-by-case basis. The authors assume full responsibility for the content and correctness of the paper. 

\clearpage

\bibliography{article}

\newpage
\appendix 
\section{APX-hardness of OPE}
\label{appA}
In \S\,\ref{sec2} we showed that the OMS problem is APX-complete. We now show the same for the the more general OPE problem (where partial merges are also allowed).

We follow the same reduction as in \S\,\ref{sec2}, adding a necessary extra step in the proof. Recall that given a cubic, undirected, unweighted graph $G$, we construct a string $s$ and an integer $k$. We now take $(s,k)$ to be an OPE instance. We claim that $G$ has a cut of size at least $c$ if and only if $\OPT(s,k) \geq 30n+c$, where $n$ is the number of vertices in $G$.

The forward direction is identical to the proof in \S\,\ref{sec2}, as the constructed OMS merge sequence also serves as a partial merge sequence for OPE. 

For the reverse direction, consider a partial merge sequence $\cR^*$ of length $k=n+1$, of utility $30n+c$ for $s$. Lemma~\ref{claim_wf} showed that a merge sequence $\cR$ that is not \emph{well-formed} can be transformed into a well-formed one, while increasing its utility. We now extend the argument to a partial merge sequence $\cR^*$. We show that we can transform $\cR^*$ into a \emph{full merge sequence} (with no partial merges) that is well-formed, and the remainder of the proof in \S\,\ref{sec2} can go through unchanged. 

\begin{lemma}
If $\cR^*$ is a partial merge sequence of length $n+1$, then we can find (in polynomial time) a well-formed (full) merge sequence $\cR$ of the same length and at least the same utility. 
\end{lemma}

\begin{proof}
As pairs other than $\#\#$ occur at most $10$ times in $s$, we assume $\cR^*$ (at least partially) merges $\#\#$, for otherwise the total utility would be at most $10(n+1) < 30n$, and any well-formed merge sequence $\cR$ would have larger utility. 

Substrings of $s$ containing $\,|\,$ occur at most once, and substrings of length $3$ occur at most $6$ times; it follows that any (partial) merge involving $\,|\,$ or a newly created symbol has utility at most $6$. Suppose there are $t$ such partial merges in $\cR^*$. Remove them from $\cR^*$ from right to left, and notice that all remaining merges in $\cR^*$ can only involve pairs of the form $\#\#$, $\ell_i\#$, and $\#\ell_i$. By Claim~\ref{lem_removal}, each removal decreases utility by at most $6$.  
Add to the end of $\cR^*$, $t$ partial merges for pairs $\ell_i\#$, where no other merge involves $\ell_i$, selecting the four safe occurrences in the four copies of $s_i$ and one safe occurrence in each incident edge gadget. Each added merge has utility at least $7$, so total utility increases.

We now only have merges of the form $\#\#$, $\ell_i\#$, $\#\ell_i$ in $\cR^*$, but some may be partial, and for some $\ell_i$, we may not have any of the two merges. If there are duplicate merges in $\cR^*$, i.e., more than one merge of the same pair, then we extend the first such merge to also perform the replacements of the other merges of the same pair. After this, we can remove the later duplicate merges (as they no longer have any utility), and replace them by some full merge $\ell_j \#$ where no other merge involves $\ell_j$. 

We then move $\#\#$ to be the last merge in $\cR^*$. Note that a partial merge sequence where all merges involve only input symbols (as the case here) can be freely re-arranged without changing utility.  
After we have moved $\#\#$ to the end of $\cR^*$, we can turn it into a full merge, as this can only increase the utility of $\cR^*$.

We now only have partial merges of the forms $\#\#$, $\ell_i\#$, and $\#\ell_i$.

Suppose that for some vertex $v_i$ both $\#\ell_i$ and $\ell_i\#$ merges appear. Suppose $\#\ell_i$ appears first;
the other case is symmetric. Remove the selected replacements of $\ell_i \#$. In each copy of $s_i = \ell_i \# \ell_i |$, a removed replacement can be compensated by an additional replacement of $\#\ell_i$. In the
edge gadgets there are at most six removed replacements, which we count as loss. Replace the
deleted merge by a partial merge $\ell_q \#$ for a previously missing vertex $v_q$, selecting the $7$ safe occurrences.
The overall gain is at least one. Repeating this leaves exactly one merge for every vertex. 
Finally process the merges from right to left and make each full. Within vertex-gadgets, this can only
increase utility. Within an edge gadget, a newly added replacement cannot overlap a later vertex
replacement and may only destroy a later $\#\#$ replacement, in which case the change has no effect on the utility. Hence, utility never decreases. The final sequence is full and well-formed.

\end{proof}

\section{A below-1/2 upper bound on the approximation ratio of BPE}
\label{appB}

\providecommand{\OPT}{\mathrm{OPT}}
\providecommand{\bpe}{\mathrm{BPE}}

The construction given in \S\,\ref{sec3} stops at the ratio $1/2$. Intuitively, BPE is ``tricked'' into merging a pair whose two symbols each are next to a merge of the
optimum; we destroy two units of utility and gain one. To improve the factor $1/2$, one must necessarily leave a constant fraction of the optimum merges intact, which the optimum can still pick up after BPE has finished its budget of merges.  
The family of instances we give below achieves this, with a worst-case ratio of 
$({1-e^{-2}})/{2} = 0.432332\ldots$.

The smallest instance in the family that beats $1/2$ has string length $|s| = 185\,484$, with $k = 1248$ merges, attaining the ratio $0.4993$; larger parameters yield instances with ratio as close
to $(1-e^{-2})/2$ as required. The constructions are clearly of astronomical size, and we made no attempt at optimizing the parameters. Unlike the earlier constructions, we need no assumption on
tie-breaking.

\subparagraph{Sketch of the construction.}
The input string contains ``traps'' arranged in \emph{levels} $M, M-1, \dots, m+1$.
Every pair that the optimum wants to merge has $MT$ initial occurrences, and in each level, BPE removes
exactly $T$ of them, so after completing the levels from $M$ to $m$ a fraction $m/M$ of the optimum is
intact.
A single trap merge at level $h$ removes $T$
occurrences from $2h$ different pairs at once, so a level is ``cleared'' using 
$\approx 1/{2h}$ of the budget; the descent from $M$ to $m$ therefore costs a
$\tfrac12(H_M-H_m)$ fraction of the $k$ merges, where $H_t$ denotes the $t$-th harmonic number.
Since only $k$ merges are available, the descent can continue only while $H_M-H_m\le 2$, i.e.,
it halts at $m \approx Me^{-2}$.
At that point an $e^{-2}$ fraction of the optimum is still intact, but BPE is done with its merges. 
The utility until that point is $\approx hT$ per merge, and level 
$h$ takes up a 
$1/2h$ fraction of the budget, so BPE collects $\approx \sum_{h=m+1}^M{\frac{k}{2h}hT} = kT(M-m)/2$ utility, about half of what it destroyed.

\subparagraph{Detailed construction.}
Fix integers $M > m \ge 1$, and $T \ge 1$, so that
$$\Delta ~:=~ \sum_{h=m+1}^{M}\frac1h ~\le~ 2,
\qquad L ~:=~ M!,
\qquad N ~:=~ \sum_{h=m+1}^{M}\frac{L}{h} = L\Delta .$$
Let $I := \{(h,r) ~:~ m < h \le M,~ 0 \le r < L/h\}$; note $|I| = N$, and fix a bijection
$\delta : I \to \mathbb{Z}_N$.  All symbol indices $u$ are read modulo $N$.
We use symbols $p_{u,\ell}$, $a_u$, $b_u$, $q_{u,\ell}$ for $u \in \mathbb{Z}_N$ and
$0 \le \ell < L$, and, as before, $|$ denotes a distinct symbol in each occurrence.
The string $s$ is the concatenation, in any order, of the following blocks:
\begin{enumerate}
\item \emph{trap blocks:~~} $T \text{ copies of } ~p_{u,\ell}\;a_u\;b_{u+\delta(h,r)}\;q_{u+\delta(h,r),\,\ell}\,|,
  \\ \quad \quad ~~~~~\text{for } (h,r)\in I,~ u \in \mathbb{Z}_N,~ rh \le \ell < (r+1)h$; \\
\item \emph{residual blocks:~~} $mT \text{ copies of } ~p_{u,\ell}\,a_u\,| ~\text{ and of }~ b_u\,q_{u,\ell}\,|,
\\ \qquad \qquad \text{for } u \in \mathbb{Z}_N,~ 0 \le \ell < L$; \\
\item \emph{extra blocks:~~} $1 \text{ copy of } ~a_u\;b_{u+\delta(h,r)}\,|,
  \\~~~~~ \text{for } (h,r) \in I,~ u \in \mathbb{Z}_N$. 
\end{enumerate}

Finally, set $k := 2NL$.
We call $p_{u,\ell}a_u$ and $b_uq_{u,\ell}$ the \emph{side pairs}: there are exactly
$NL + NL = k$ of them, and they are the pairs the optimum will merge; we call
$a_u b_{u+\delta(h,r)}$ the \emph{trap pair of level $h$}.

The shift by $\delta(h,r)$ is what makes the levels interact: since $\delta$ is a bijection, the
same symbol $a_u$ appears in a trap pair of \emph{every} level, so a single side pair
$p_{u,\ell}a_u$ is ``attacked'' once per level, losing $T$ occurrences each time.  

\subparagraph{Example.}
For $M=2$, $m=1$, $T=2$ we have $L=2$, $N=1$, $k=4$, and (writing $u=0$ throughout) $s$ consists
of the trap blocks $p_{0,0}a_0b_0q_{0,0}|$ and $p_{0,1}a_0b_0q_{0,1}|$, twice each, the residual
blocks $p_{0,0}a_0|$, $p_{0,1}a_0|$, $b_0q_{0,0}|$, $b_0q_{0,1}|$, twice each, and the single
extra block $a_0b_0|$.  Each of the four side pairs occurs $MT = 4$ times, while $a_0b_0$ occurs
$4+1 = 5$ times, so BPE merges $a_0b_0$ first and is left with gains $2,2,2$, for a total of
$11$ against an optimum of $4\cdot 4 = 16$.  (There is only one level here, and the
ratio $11/16$ is of course above $1/2$; the mechanism starts to pay off for larger $M$ only.)

\subparagraph*{Analysis.}

Let us first count occurrences of the different pairs. 
Fix $u$ and $\ell$.  For each level $h$ there is exactly one $r$ with $rh \le \ell < (r+1)h$,
namely $r = \lfloor \ell/h\rfloor$, and the corresponding trap blocks contribute $T$ occurrences
of $p_{u,\ell}a_u$.  Together with the $mT$ occurrences in the residual blocks, this gives
$(M-m)T + mT = MT$.  The count for $b_uq_{u,\ell}$ is the same: for each $h$ and each $\ell$
there is exactly one $u'$ with $u' + \delta(h,\lfloor \ell/h\rfloor) = u$.
The trap pair of level $h$ belonging to $(h,r)$ and $u$ occurs in the $h$ values of $\ell$ in
$[rh,(r+1)h)$ with $T$ copies each, plus once in its extra block, i.e., $hT+1$ times; there are
$N \cdot L/h$ such pairs.  Every pair containing a separator occurs once.

To bound the optimum, we construct a feasible solution, as follows. 
Merge the $k$ side pairs, in any order.  In every block the two side pairs are disjoint, and no
side merge removes an occurrence of another side pair, since $a_u$ only ever occurs as the right
element of a side pair and $b_u$ only as the left one.  Hence each of the $k$ merges has utility
$MT$ and
$$\OPT^m(s,k) ~\ge~ kMT.\eqno{(1)}$$

\subparagraph{Behavior of BPE.}
We claim BPE merges all trap pairs of level $M$, then all of level $M-1$, and so on down to
level $m+1$.
First, a merge of utility $g$ neither increases the count of an existing pair nor creates a pair
of count more than $g$: every new occurrence contains the newly introduced symbol, occurring $g$ times; all other pair types can
only lose occurrences.
Now suppose levels $M,\dots,h+1$ have been processed.  Every side pair has lost $T$ occurrences
per processed level, so its count is $MT - (M-h)T = hT$.  The pairs present are then
\begin{itemize}
\item trap pairs of level $h$, of count $hT+1$;
\item side pairs, of count $hT$;
\item trap pairs of a level $h' < h$, of count $h'T + 1 \le (h-1)T+1 \le hT$ (using $T\ge1$);
\item pairs containing a symbol created by an earlier merge, of count at most $T$: such a symbol
      is some $c = a_ub_v$, and the only pairs containing it are $p_{u,\ell}c$ and $cq_{v,\ell}$,
      of count $T$ each, and one pair with a separator;
\item pairs containing a separator, of count $1$.
\end{itemize}
The maximum count is $hT+1$, attained \emph{only} by trap pairs of level $h$.
Distinct trap pairs of the same level occur in disjoint blocks and hence do not affect each
other's counts, so the same remains true after any number of them have been merged, and the
order within a level is arbitrary.  Note that due to the extra blocks, there is no tie between the trap- and side pairs.

\smallskip

The number of trap merges is
$$R ~=~ \sum_{h=m+1}^{M} \frac{NL}{h} ~=~ NL\Delta ~=~ \frac{k\Delta}{2} ~\le~ k, \eqno(2)$$
using $\Delta \le 2$, and their total utility is
$$\sum_{h=m+1}^{M} \frac{NL}{h}\,(hT+1) ~=~ \frac{kT(M-m)}{2} + R. \eqno(3)$$

Afterwards every side pair has exactly $mT$ occurrences remaining in the residual blocks, and no pair has count above $mT$.
A single merge can remove residual occurrences of at most one side pair type, because each
residual block is isolated by separators and contains exactly one non-separator pair.
Only $k-R < k$ rounds remain while there are $k$ side pair types, so before every remaining round
some side pair type is still untouched and has count exactly $mT$.
Hence every remaining round has utility exactly $mT$, regardless of the choice of BPE: for
$m \ge 2$ the maximizers are exactly the untouched side pairs, and for $m = 1$ the pairs
$p_{u,\ell}c$ and $cq_{v,\ell}$ tie with them, but they also have count $T = mT$ and
they touch no residual block.    Therefore
$$\bpe(s,k) ~=~ \frac{kT(M-m)}{2} + R + mT\,(k-R). \eqno(4)$$

\subparagraph{Putting all together.}
Combining $(1)$, $(2)$, $(3)$, and $(4)$, we obtain:
$$\frac{\bpe(s,k)}{\OPT^m(s,k)} ~\le~ \frac12 ~+~ \frac{m}{2M}\,(1-\Delta) ~+~ \frac{\Delta}{2MT}. \eqno(5)$$
Given $\varepsilon > 0$, choose first $M$ and then $T$, as follows.  For each $M$ let $m$ be
smallest with $H_M - H_m \le 2$; then $\Delta = H_M-H_m \le 2$ as required, and since
$H_M - H_{m-1} > 2$ we get $\Delta > 2 - 1/m$, so $\Delta \to 2$ and, from
$H_M - H_m = \ln(M/m) + O(1/m)$, also $m/M \to e^{-2}$ as $M \to \infty$.
Thus the middle term of $(5)$ tends to $-e^{-2}/2$, and we fix $M$ large enough that it is below
$-e^{-2}/2 + \varepsilon/2$.  The last term is at most $1/(MT)$ and we fix $T$ with
$1/(MT) \le \varepsilon/2$.  This yields a finite instance with
$$\frac{\bpe(s,k)}{\OPT^m(s,k)} ~\le~ \frac{1-e^{-2}}{2} + \varepsilon,$$
independently of tie-breaking.  Since the construction uses full merges only, the bound applies
to OMS; as the OPE optimum is at least the OMS optimum, it applies to OPE as well.

\subparagraph*{A concrete instance below $1/2$.}
Take $M = 4$, $m = 1$, $T = 14$ and $L = 4! = 24$, so that $\Delta = 13/12$, $N = 26$ and
$k = 1248$.  The string has $185\,484$ symbols, and BPE is forced through the profile
\begin{center}
\begin{tabular}{l|c|c|c}
stage & merges & utility each & total \\ \hline
level $h=4$ & $156$ & $57$ & $8892$ \\
level $h=3$ & $208$ & $43$ & $8944$ \\
level $h=2$ & $312$ & $29$ & $9048$ \\
residual & $572$ & $14$ & $8008$ \\ \hline
& $1248$ & & $34892$
\end{tabular}
\end{center}
while the $1248$ side merges have utility $56$ each, i.e.\ $\OPT^m(s,k) \ge 69888$.  Hence
$$\frac{\bpe(s,k)}{\OPT^m(s,k)} ~\le~ \frac{34892}{69888} ~=~ \frac{671}{1344} ~=~ 0.499256\ldots ~<~ \frac12 ,$$
so already this instance shows that BPE is not a $2$-approximation.

\end{document}